\documentclass[letterpaper, 10 pt, conference]{ieeeconf}  

\IEEEoverridecommandlockouts                              

\usepackage{graphicx} 
\usepackage{bm}
\usepackage{amsmath} 
\usepackage{amssymb}  
\usepackage{float}
\usepackage{cite}

\newcommand{\deltab}{\bm{\delta}}

\newcommand{\Eb}{\mathbf{E}}

\newcommand{\Ub}{\mathbf{U}}
\newcommand{\Xb}{\mathbf{X}}
\newcommand{\Wb}{\mathbf{W}}
\newcommand{\Zb}{\mathbf{Z}}

\newcommand{\IC}{\mathcal{I}}

\newcommand{\RC}{\mathcal{R}}

\newcommand{\NC}{\mathcal{N}}

\newcommand{\eq}{&\hspace{-0.5em}=\hspace{-0.5em}&}
\newcommand{\tr}{{\rm{tr}}}
\newtheorem{theorem}{Theorem}
\newtheorem{definition}{Definition}
\newtheorem{corollary}{Corollary}
\newtheorem{lemma}{Lemma}
\newtheorem{remark}{Remark}
\newtheorem{proposition}{Proposition}

\usepackage{xcolor}

\title{\LARGE \bf
Value of Information in Networked Control Systems Subject to Delay
}

\author{ Siyi Wang, Qingchen Liu, Precious Ugo Abara, John S. Baras and Sandra Hirche
\thanks{*This work was  funded by the German Research Foundation (DFG) under the grant number 315177489 as part of the SPP 1914 (CPN). The work of Q. Liu was in addition supported by the European Union’s Horizon 2020 research and innovation programme under the Marie Skłodowska-Curie Grant 754462. The work of J.~S.~Baras was in addition partially supported by ONR grant N00014- 17-1-2622.} 
\thanks{
 Siyi Wang, Qingchen Liu, Precious Ugo Abara and Sandra Hirche are with the Chair of
Information-oriented Control (ITR), Technical University of Munich, Germany, 
        {\tt\small         \{siyi.wang, qingchen.liu, ugoabara, hirche\}@tum.de}       }
\thanks{
John S. Baras is with the Department of Electrical
\& Computer Engineering, Institute for Systems Research, University of
Maryland, USA, 
        {\tt\small \{baras\}@umd.edu}
        }%
}

\begin{document}

\maketitle
\thispagestyle{empty}
\pagestyle{empty}


\begin{abstract}
In this paper, we study the trade-off between the transmission cost and the control performance of a networked control system subject to  network-induced delay. 
Within the linear–quadratic–Gaussian (LQG) framework, the joint design of control policy and networking strategy is decomposed into separate optimization problems. Based on the trade-off analysis, a  delay-dependent Value-of-Information (VoI)    metric which quantifies the value of transmitting a data packet is introduced. The VoI enables the decision-makers embedded in subsystems to design the triggering policy. The proposed scalable VoI inherits the task criticality of the existing VoI metric. Additionally, the sensitivity to the system parameters such as information freshness and network delays is directly derivable. The VoI-based scheduling policy is shown to outperform the periodical triggering policy and the  Age-of-Information (AoI) based policy for network control systems under transmission delay. The effectiveness of the constructed VoI with arbitrary network delay is validated through numerical simulations.

\end{abstract}

\section{INTRODUCTION}

Networked control systems (NCSs) are generally referred to multiple or even a large number of plants that are controlled by computational algorithms and supported by a wired or wireless communication network providing information exchange \cite{antsaklis2007special,walsh2001scheduling, park2017wireless}. Their application domains are multi-fold, including for example smart energy grids \cite{persis2018power}, robotic systems \cite{bullo2009distributed} and autonomous production lines \cite{scholz2009modelling}. From a theoretical perspective, NCSs can be seen as realizations of a scenario in which multiple feedback control loops are closed over a shared communication network. In data scheduling of NCSs, the event-triggered schemes impel a data transmission  only when a pre-designed  triggering condition is satisfied. A plethora of works e.g. \cite{lunze2010state, wang2011event}, have shown that event-triggered schemes performs well in reducing the network communication resource consumption while guaranteeing similar control performance compared to widely-used time-triggered scheme \cite{heemels2012introduction}. In NCSs, two main layers (control layer and communication layer) strongly influence the performance of each other and face heterogeneous inter-layer couplings such as network-induced delay, packet loss and quantization. 
The inevitable presence of transmission delay in NCSs indeed influences the system performance, which has drawn significant attention in recent studies \cite{heemels2010networked,yue2013delay,liu2019survey}.\\  
\indent 
Age-of-Information (AoI) is a recently introduced metric which is capable of coordinating communication resource allocation by minimizing the information freshness of each subsystem. AoI captures the information freshness at the controller and is scalable for application in large-scale systems \cite{kosta2017age}. By scalable we mean that the AoI is computationally tractable and implementable on a NCS with a large number of subsystems. However, as shown in \cite{champati2019performance, ayan2019age}, minimizing average information freshness does not necessarily lead to individual control performance satisfaction, and might also result in an undesirable collective performance.  \\
\indent 
Compared to AoI in NCSs,  Value-of-Information (VoI) is quantified as the
variation in a  value
function with respect to a piece of information about the state of the
process available to the controller  \cite{soleymani2021value,molin2019scheduling}. In \cite{soleymani2021value}, it has been shown that the VoI-based scheduling policy is an optimal policy, by assuming that network-induced effects, such as transmission delay, quantization or packet dropouts, are negligible. However, this assumption is definitely not valid in real-world applications. Moreover, another disadvantage of the previous VoI in \cite{soleymani2021value} is its computational complexity -- in particular in comparison to AoI. \\
\indent 
In this paper, we extend the VoI concept to NCSs subject to arbitrary transmission delay. Firstly, under some mild assumptions on the triggering policy, the initial joint design problem is decomposed into two sub-problems $i)$ the optimal control policy design, and $ii)$ the optimal triggering policy design. By addressing the trade-off between LQG performance and communication cost and assuming that the triggering decisions only depend on primitive random variables, the solution to the optimal control problem is shown to be the certainty equivalent controller. Analyzing the optimal triggering policy design problem, we obtain a VoI metric which  is described as a measure of the urgency of data transmission between the event-trigger and the remote controller. The VoI serves as the triggering condition for NCS closed over a communication network. 
Our constructed  VoI function sufficiently accounts for individual task requirements of networked control loops as previous VoI while capturing the influence induced by transmission delay. We prove that the delay-dependent VoI-based scheduling policy is optimal, and improves the LQG control performance compared to both periodic scheduling and AoI-based scheduling policies. To reduce the  computation burden in solving the scheduling problem using dynamic programming, we provide a simplified VoI proxy function based on an approximation approach. \\
\indent 
The remainder of this paper is structured as follows. Section \ref{sec:preliminaries} introduces the system model and the problem formulation. Section \ref{sec:main result} presents the main result on constructing delay-dependent VoI and its performance analysis. Section \ref{sec:simulation}  illustrates our results through simulation examples.  \\
\indent
\textbf{Notations:}  In this study, $\Eb[\cdot]$ and $\Eb[\cdot\mid\cdot]$ denote the expected value and the conditional expectation, respectively. Let $x \sim \NC(\mu, C_{x})$ represent Gaussian random variable $x$ with mean $\mu$ and covariance matrix $C_{x}$. A sequence of a random process $\{x_{k}\}_{k \ge 0}$ is denoted as $\Xb_{k}=\{ x_{0}, \dots, x_{k} \}$ for its complete history  up to time index $k$.
\section{PRELIMINARIES}\label{sec:preliminaries}

\subsection{System model}
We consider a networked control system in which the feedback loop is coupled through a communication network. The stochastic discrete-time process to be controlled is described by a discrete stochastic difference equation, which is assumed to be linear and time-invariant. The stochastic difference equation is given by 
\begin{eqnarray}\label{eq:plant}
x_{k+1}  \eq A x_{k}  +B u_{k}  +w_{k} 
\end{eqnarray}
where system state $x_k \in \mathbb{R}^{n}$, control signal $u_{k}\in\mathbb{R}^m$, $A\in \mathbb{R}^{n\times n}$ and $B \in \mathbb{R}^{n\times m}$. The variables $x_k$ and $u_k$ denote the system state and the control input, respectively. The system matrix pair $(A,B)$ is assumed to be controllable. The process noise $w_k \sim \NC(0, W)$ is independent identically distributed (i.i.d) with zero mean and  postive semidefinite covariance. The initial state  $x_0 \sim \NC(0, R_{0}) $  is a random vector with zero mean and finite covariance. The random variables $x_0$ and $w_{k}$ are assumed to be statistically independent for each $k$. Across this paper, we call $x_0$ and $w_k$ the primitive random variables of the system.

\subsection{Network model and information structure}

The control system is equipped with a scheduler (event-trigger), to determine whether allows the state $x_{k} $ to be transmitted via the communication network at each time instant $k$. See Fig.~\ref{fig:system architecture} for a graphical illustration. 
The transmission triggering variable $\delta_{k}$ of the scheduler takes value from $\{0, 1\}$, and is given by
\begin{eqnarray}\label{eq:triggering policy}
   \delta_{k} = \gamma (\IC_{k}^{e}) =
   \left\{ 
    \begin{array}{ll}
        1 &  {\rm{transmission~occurs}}\\
         0& {\rm{otherwise}}
    \end{array}\right.
\end{eqnarray}
where $\gamma(\cdot)$ denotes the triggering policy of the control system and $\IC_{k}^{e}$ denotes the information set for the event-trigger. Its formal definition is provided in \eqref{eq:event-trigger infor}.  
\begin{figure}[t]
    \centering
    \includegraphics[width=0.45\textwidth]{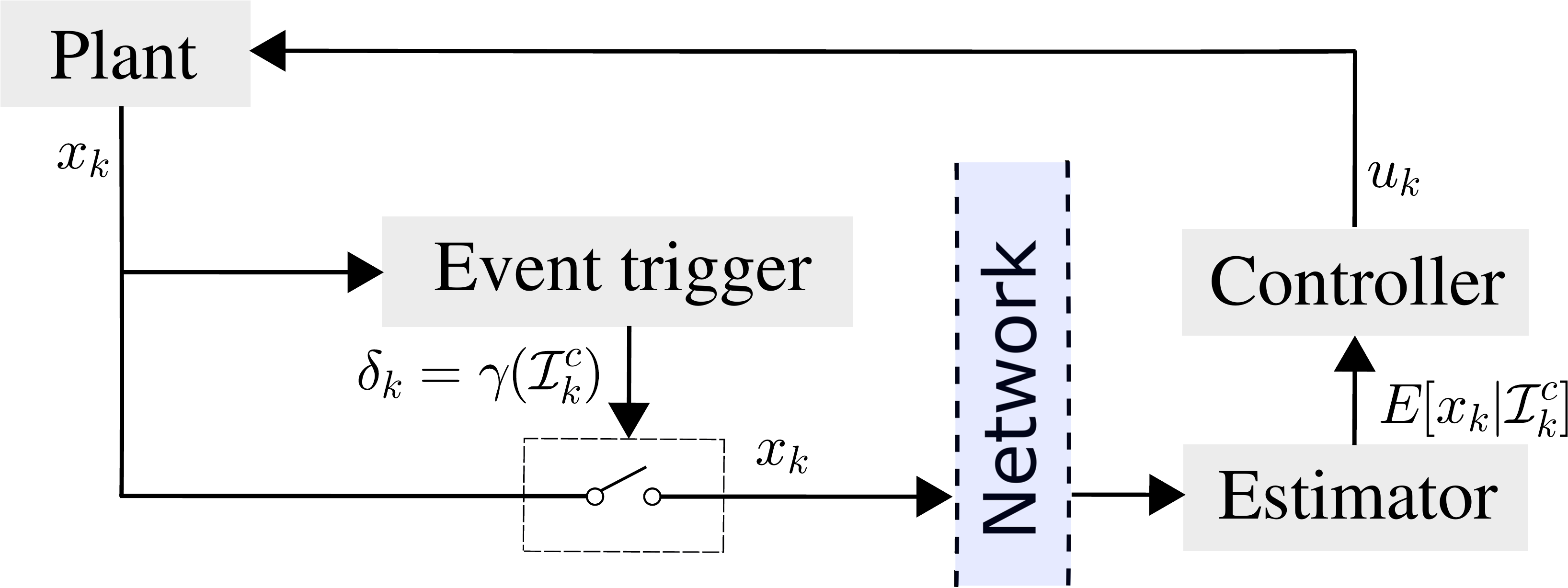}
    \caption{A NCS with shared communication channel.}
    \label{fig:system architecture}
\end{figure}
We assume that the data transmission is subject to a $\tau $-step delay (see Definition~\ref{definition:AoI}), which is induced by the communication network. To capture this delay, we assume that the controller keeps a track of the time instant it received the latest packet. This information updating time instant is a function of the triggering variables. 
By using $ \delta_{k}$, the information updating time at the controller side is described by $c(k) = \max \{ t \vert \delta_{t-\tau} =1, t \le k \}$. The transmission timeline in the networked control system is illustrated in Fig.~\ref{fig:transmission timeline}. To quantify the information freshness of the data packet at the controller side, we introduce the concept of Age-of-Information (AoI) as follows.  
\begin{definition}\label{definition:AoI}
The delay-dependent AoI  is defined as  
\begin{eqnarray}
\Delta_{k} = k- s(k)  
\end{eqnarray} 
where 
\begin{eqnarray}\label{eq:sk}
 s(k) = c(k)-\tau
\end{eqnarray} 
is the generating time instant of the latest information update of the controller.
\end{definition}

At time instant $k$, the controller computes its action based on the information updates. At the sensor side, the event-trigger determines whether to transmit the current state $x_{k}$ to the network according to the scheduling policy $\gamma_{k}(\IC_{k}^{e})$. We will detail the design of our VoI-based policy in Section III. The information received at the controller side, at time instant $k$, is denoted as   
\begin{eqnarray*}
z_{k} = \left\{ \begin{array}{ll}
x_{s(k)}     &  {\rm{if}} ~\delta_{k-\tau} =1,\\
\emptyset     & {\rm{otherwise,}}
\end{array}\right.
\end{eqnarray*}
where the time index $s(k)$ is defined in \eqref{eq:sk}, and defining $\delta_{-\tau}=\ldots=\delta_{-1}=0$. We denote the observation history of the controller until time $k$ as $\Zb_{k}$. The information set at the controller side is denoted as 
\begin{eqnarray}\label{eq:controller infor}
\IC_{k}^{c} = \{ \Zb_{k},  \Ub_{k-1} \},
\end{eqnarray}
with the initial information $\IC_{0}^{c} = \{z_{0} \}$, and $\Ub_{k-1}$ denotes the control history up to time instant $k-1$. 
The control history $\Ub_{k-1}$ up to time instant $k-1$ is also available to the event-trigger. It follows that, the information set at the event-trigger is denoted as 
\begin{eqnarray}\label{eq:event-trigger infor}
\IC_{k}^{e} = \{ \Xb_{k},  \Ub_{k-1}, \deltab_{k-1} \}, 
\end{eqnarray}
where  $\Xb_{k}$ and $\deltab_{k-1}$ denote the system state history up to time $k$ and triggering decision history up to time $k-1$, respectively. Obviously the initial information set $\IC_{0}^{e} = \{x_{0},\delta_{0}\}$ because of the empty control history.
\begin{figure}[t]
    \centering
    \includegraphics[width=0.45\textwidth]{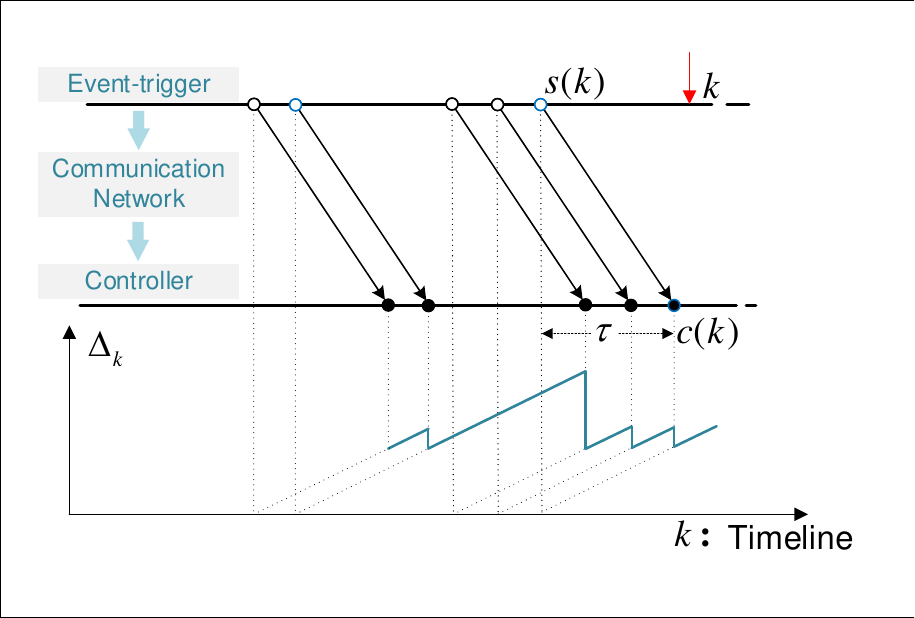}
    \caption{Transmission timeline} 
    \label{fig:transmission timeline}
\end{figure}

\subsection{Problem statement} 
In this study, we aim to find the optimal joint co-design of control and scheduling policy which solves the following optimization problem: 
\begin{eqnarray} \label{eq:original opt} 
     \min_{ f, \gamma}   \Psi(f,\gamma) =  J(f, \gamma)  + \RC(\gamma),
\end{eqnarray}
where  $f$ denotes the control policy and $\gamma$ denotes the scheduling policy (defined in \eqref{eq:triggering policy}),
$J(f, \gamma)$ and  $\RC(\gamma)$  represent the control cost function and communication cost function, respectively. The control cost function is given by
\begin{eqnarray}\label{eq:LQG}
J = \frac{1}{T+1}\Eb\big[ \sum_{k=0}^{T } \big(\hspace{-1em}&x_{k}^\top Q_{k}x_{k} + u_{k}^\top R_{k} u_{k}\big) \nonumber \\
& + x_{T+1}^\top Q_{T+1}x_{T+1}  \big],  
\end{eqnarray} 
where the matrix $Q_{k}$ is positive semidefinite,  $R_k$ is positive definite, the matrix pair $(A, Q_k^{\frac{1}{2}})$ is detectable, with $Q_k = \left(Q_k^{\frac{1}{2}}\right)^\top Q_k^{\frac{1}{2}}$. The communication cost function we considered is designed as:
\begin{eqnarray*}
\RC \eq \frac{1}{T+1}\Eb\big[\sum_{k=0}^{T}\theta_{k} \delta_{k} \big],
\end{eqnarray*}
where $\theta_{k} $ denotes the single transmission cost at time instant $k$. 

\section{Main Result} \label{sec:main result}
Before presenting the main result in this work, we provide the following supporting result. \\
\indent As proved in \cite{molin2013on}, a pair of joint control and communication policies with a certainty equivalence controller is a dominating class of policies.
Thus, within the dominating policies, the local sub-problems can be decomposed into: $i)$ the design of certainty equivalence controller; and $ii)$ the design of optimal triggering policy \cite{molin2014price}. 
In the following, we provide a condition on the scheduler to guarantee
certainty equivalence of the controller. 

\begin{lemma}[\cite{ramesh2011on}]\label{lemma:certainty eq ctrl}
Let the triggering policy be a function of the primitive random variables, i.e., 
$\delta_k  = \gamma_k \left( x_0, \Wb_{k-1}\right)$.    
Then the optimal controller which minimizes the finite horizon LQG problem \eqref{eq:LQG} is certainty equivalent, i.e.,
\begin{eqnarray}\label{eq:controller}
u_{k}^{\ast}= f^{\ast}( \IC_{k}^{c})= L_{k} \Eb[x_{k}\mid \IC_{k}^{c}],
\end{eqnarray}
with $L_{k} = -(\Lambda_{k})^{-1}B^{\top}P_{k+1}A $ and $P_{k}$ being the solution of the following algebraic Riccati equation:
\begin{eqnarray}\label{eq:Riccati matrix} 
P_{k} \eq  Q_{k}+A^{\top}(P_{k+1}-P_{k+1}B(\Lambda_{k})^{-1}B^{\top}P_{k+1})A,  \nonumber \\
P_{T} \eq  Q_{T}, \quad \Lambda_{k} =  R_{k}+B^{\top}P_{k+1}B.
\end{eqnarray} \hfill $\blacksquare$
\end{lemma}
Lemma~\ref{lemma:certainty eq ctrl} decomposes the co-design problem \eqref{eq:original opt}, and therefore the remaining problem is how to find the optimal triggering policy, which is detailed below.
For notation simplicity, the subscript $i$ is omitted from the next subsection. 

\subsection{System dynamics reparameterization}
We design the estimator at the controller-side as
\begin{eqnarray}\label{eq:estimator}
  \Eb[x_{k}  \mid \IC_{k}^{c}]  =  A^{\Delta_{k}}x_{s(k)} + \sum_{r=1}^{\Delta_{k}}A^{r-1}Bu_{k-r},
\end{eqnarray}
where $u_{k}$ is defined in \eqref{eq:controller} and $s(k)$ is defined in \eqref{eq:sk}. The estimation error is defined as  
\begin{eqnarray}\label{eq:est error}
e_{k}= x_{k} - \Eb[x_{k}\mid \IC_{k}^{c}]
\end{eqnarray}
Similarly, starting from the latest triggering time instant $s(k)$, the state dynamics \eqref{eq:plant} can be further written as  \begin{eqnarray}\label{eq:re plant}
      x_{k} \eq  A^{\Delta_{k}}x_{s(k)} + \sum_{r=1}^{\Delta_{k}}A^{r-1}\big(Bu_{k-r} +w_{k-r}\big).
\end{eqnarray}
By substituting \eqref{eq:estimator} and system dynamics \eqref{eq:re plant} into the error function \eqref{eq:est error}, the estimation error is obtained as  
\begin{eqnarray}\label{eq:error repara}
    e_{k} =  \sum_{r=1}^{\Delta_{k}}A^{r-1}w_{k-r},
\end{eqnarray}
which will be served as a key component in designing the VoI-based scheduling policy.\\
\indent
In order to characterize how the triggering policy influences the control performance, the estimation error \eqref{eq:error repara} is reparameterized with respect to the triggering variables. 
Note that the value of $\Delta_{k+\tau}$ and error dynamics \eqref{eq:error repara} are determined by the decisions up to time instant $k$. 
The AoI expression with respect to decision $\delta_{k}$ at time instant $k+\tau$ is given as 
\begin{eqnarray*} 
    \Delta_{k+\tau} = (1-\delta_{k})(\Delta_{k+\tau-1}+1) +\delta_{k}\tau.
\end{eqnarray*}
For illustrative purpose, we provide an example. If the triggering occurs at $k-1$, while not at $k$, then $\Delta_{k+\tau}=\tau+1$. Besides, if no triggering occurs during $[k-\tau, k]$, then $\Delta_{k+\tau} =\Delta_{k}+\tau$. 
Accordingly, the expression of $e_{k+\tau}$ with respect to decision $\delta_{k}$ is denoted as
\begin{eqnarray}\label{eq:one step Gaussian}
       e_{k+\tau}  \eq  (1-\delta_{k}) (Ae_{k+\tau-1} +  w_{k+\tau-1})   \nonumber \\
       &&+ \delta_{k}\sum_{r=1}^{\tau}A^{r-1}w_{k+\tau-r} \nonumber\\
       \eq \sum_{r=1}^{\Delta_{k+\tau}}A^{r-1}w_{k+\tau-r}.
\end{eqnarray}

\subsection{VoI Construction}
We formally define the VoI as follows 
\begin{eqnarray}
\label{eq:VoI1}
    {\rm{VoI}}_{k}:=V_{k}(\IC_{k}^{e})\mid_{\delta_{k}=0} - V_{k}(\IC_{k}^{e})\mid_{\delta_{k}=1}.
\end{eqnarray}
Equation \eqref{eq:VoI1} indicates the value  that is assigned to the reduction of uncertainty from the decision maker’s perspective given a measurement update \cite{soleymani2021value}, where
$V_{k}(\IC_{k}^{e})$ denotes the value function in the optimization problem from the viewpoint of the event-trigger. Now we are ready to present our main result. 
\begin{theorem}\label{theorem:VoI}
Let the triggering policy be a function of the primitive random variables, i.e., $\delta_k = \gamma_k\left( x_0, \Wb_{k-1}\right)$. And let the optimal control policy be given by \eqref{eq:controller}. The optimal scheduling policy which minimizes the optimization problem \eqref{eq:LQG} is the VoI-based policy given by
\begin{eqnarray}\label{eq:VoI}
 \delta_{k}^{\ast} = \mathbf{1}_{{\rm{VoI}}_{k}> 0 } =      \left\{  \begin{array}{ll}
    1    & {\rm if}~{\rm VoI}_{k}>0,  \\
     0   & {\rm otherwise,}
   \end{array} \right. 
\end{eqnarray} 
where ${\rm{VoI}}_{k}$ is the VoI at time instant $k$ expressed as
\begin{eqnarray}\label{eq:VoI expression}
{\rm{VoI}}_{k}  =  -\theta_{k} +   \tr(\Gamma_{k+\tau}\Phi(\Delta_{k+\tau})    +\rho_{k},
\end{eqnarray} 
where $\Phi(\Delta_{k+\tau}) = \Eb[  \sum_{r=\tau+1}^{\Delta_{k+\tau-1}+1}(A^{r-1})^{\top}A^{r-1}\|w_{k+\tau-r}\|^{2}] $ with  $\Gamma_{k} = L_{k}^{\top}\Lambda_{k}L_{k}$.  Besides, the $\rho_{k}$ is expressed as 
\begin{eqnarray}\label{eq:rho}
 \rho_{k} =  \Eb[V_{k+1} \mid \IC_{k}^{e},~\delta_{k}=0 ] - \Eb[V_{k+1} \mid \IC_{k}^{e},~\delta_{k}=1 ] 
\end{eqnarray}
with $V_{k}  =  \Eb[\sum_{t=k}^{T-\tau} \theta_{t}\delta_{t} +    e_{t+\tau}^{\top}\Gamma_{t+\tau}e_{t+\tau}   \mid \IC_{k}^{e} ]$.  
\end{theorem}
\begin{proof}
By substituting the algebraic Riccati equation \eqref{eq:Riccati matrix} into the optimization cost function \eqref{eq:original opt}, we can rewrite \eqref{eq:original opt} as 
\begin{eqnarray}
\label{eq:LQG recursive}
    \Psi (f,\gamma ) = \frac{1}{T+1}\hspace{-2.5em}& \Eb[x_{0}^{\top}P_{0}x_{0} + \sum_{k=0}^{T}  w_{k}^{\top}P_{k+1}w_{k} \nonumber \\
    &+(u_{k}+L_{k}x_{k})^{\top}\Gamma_{k}(u_{k}+L_{k}x_{k}) + \theta_{k}\delta_{k} ] 
\end{eqnarray}
Inserting the certainty equivalence controller \eqref{eq:controller} into the cost function \eqref{eq:LQG recursive} results in
\begin{eqnarray} \label{eq:LQG error}
\Psi(f^{\ast},\gamma)  \eq  \frac{1}{T+1}
 \Eb[x_{0}^{\top}P_{0}x_{0} +\sum_{k=0}^{T}   w_{k}^{\top}P_{k+1}w_{k}\nonumber \\ &&\hspace{7.3em}  +   e_{k}^{\top}\Gamma_{k}e_{k}+ \theta_{k}\delta_{k}].
\end{eqnarray}
where $\Gamma_{k} = L_{k}^{\top}\Lambda_{k}L_{k}$. 
We define a cost-to-go function at time $k$ for \eqref{eq:LQG error} as 
\begin{eqnarray}\label{eq:Vk}
V_{k}  =  \Eb[\sum_{t=k}^{T-\tau} \theta_{t}\delta_{t} +    e_{t+\tau}^{\top}\Gamma_{t+\tau}e_{t+\tau}   \mid \IC_{k}^{e} ],
\end{eqnarray}
where $e_{t}$, for $0 \le t < \tau $, is independent of the triggering policy since no packet arrives at the controller during $t \in [0,\tau)$ due to the transmission delay. It is straightforward to re-write $V_{k}$ as 
\begin{eqnarray} \label{eq:Vk 1}
V_{k} =  \Eb[ \theta_{k}\delta_{k} +    e_{k+\tau}^{\top}\Gamma_{k+\tau}e_{k+\tau} +    V_{k+1}  \mid \IC_{k}^{e}].
\end{eqnarray} 
The minimizer $\delta_{k}^{\ast}  $ of cost-to-go function \eqref{eq:Vk 1} is obtained as $\delta_{k}^{\ast} = \mathbf{1}_{{\rm{VoI}}_{k}> 0 }$ with ${\rm{VoI}}_{k} $ defined in \eqref{eq:VoI1}. 
Regarding the term $\Eb[    e_{k+\tau}^{\top}\Gamma_{k+\tau}e_{k+\tau}   \mid \IC_{k}^{e}]$ in stage cost of $V_{k}$,  
by substituting \eqref{eq:one step Gaussian}, it can be further written as  
\begin{eqnarray}\label{eq:stage cost phi}
&&\Eb[
 e_{k+\tau}^{\top}\Gamma_{k+\tau}e_{k+\tau} \mid \IC_{k}^{e},~\delta_{k}=0 ] \nonumber \\
 && -  \Eb[
 e_{k+\tau}^{\top}\Gamma_{k+\tau}e_{k+\tau} \mid \IC_{k}^{e},~\delta_{k}=1 ] \nonumber \\ \eq     \sum_{r=1}^{\Delta_{k+\tau-1}+1}(A^{r-1}w_{k+\tau-r})^{\top}\Gamma_{k+\tau}A^{r-1} w_{k+\tau-r}   \nonumber 
 \\ && -  \sum_{r=1}^{\tau}(A^{r-1}w_{k+\tau-r})^{\top}\Gamma_{k+\tau}A^{r-1} w_{k+\tau-r}   \nonumber \\
\eq  \tr(\Gamma_{k+\tau}\Phi(\Delta_{k+\tau})),
\end{eqnarray}  
with $\Phi(\Delta_{k+\tau})$ defined in \eqref{eq:VoI expression}. Overall, the combination of \eqref{eq:VoI1}, \eqref{eq:Vk 1} and \eqref{eq:stage cost phi} implies the expression of VoI
\begin{equation}
{\rm{VoI}}_{k}  =  -\theta_{k} +   \tr(\Gamma_{k+\tau}\Phi(\Delta_{k+\tau})    +\rho_{k}
\end{equation}
 with $ \rho_{k} =  \Eb[V_{k+1} \mid \IC_{k}^{e},~\delta_{k}=0 ] - \Eb[V_{k+1} \mid \IC_{k}^{e},~\delta_{k}=1 ] $. 
\end{proof}
To obtain a valid VoI function \eqref{eq:VoI expression}, we need to minimize the cost-to-go function  \eqref{eq:Vk} using dynamic programming. In the numerical backward induction of dynamic programming, the objective function must be computed for each combination of values. The cost-to-go function $V_{k}$ \eqref{eq:Vk} is dependent on future triggering variables ($\delta_{k+1},\delta_{k+2},\dots,\delta_{N}$) and needs to be computed
recursively, therefore, its computation is in general difficult. In order to achieve generic computational tractability and scalability, we introduce an approximation of the VoI function using Proposition 4 of \cite{soleymani2021value}, a so-called VoI proxy function. \\
\indent 
Let $\bar{\gamma} = \{\bar{\delta}_{0}, \dots, \bar{\delta}_{N} \}$ be a periodic triggering policy with $\bar{\delta}_{k} = 1$ for all $0 \le k \le N$. At each time instant $k$, we use the  sub-optimal set of policies  $\{\delta_{k}, \bar{\delta}_{k+1}, \dots, \bar{\delta}_{N}\} $ as the baseline policy to 
simplify the $\rho_{k}$ calculation, and it results in $\rho_{k} =0$.

\begin{proposition}\label{proposition:sub-optimal}
Consider the joint optimization problem \eqref{eq:LQG} and let the optimal control policy be given by \eqref{eq:controller}. The VoI proxy-based scheduling policy $\tilde{\gamma}$ is given by
\begin{eqnarray}\label{eq:aproxy VoI}
 \tilde{\delta}_{k} = \mathbf{1}_{{\rm{VoIP}}_{k}> 0 },
\end{eqnarray}
where
\begin{eqnarray}\label{eq:aproxy VoI expression}
{\rm{VoIP}}_{k} =   -\theta_{k} +   \tr(\Gamma_{k+\tau}\Phi(\Delta_{k+\tau}) )
\end{eqnarray} with $\Phi(\Delta_{k+\tau})$ defined in \eqref{eq:VoI expression}. The policy is sub-optimal but outperforms the periodical policy in tradeoff between the control performance and communication cost.  

\end{proposition}  
\begin{proof}
We need to prove $ \Psi(\tilde{\gamma}, f^{\ast}) \le \Psi(\bar{\gamma}, f^{\ast})$.
Let the cost-to-go function under sub-optimal triggering policy and periodical policy be $\tilde{V}_{k}$ and $\bar{V}_{k}$, respectively. In order to show $ \Psi(\tilde{\gamma}, f^{\ast}) \le \Psi(\bar{\gamma}, f^{\ast})$, it is enough to show  $V_{k} \le \tilde{V}_{k}$.  Assume that the claim holds for $k+1$, we have 
\begin{eqnarray*}
\tilde{V}_{k}  \eq \Eb[\theta\tilde{\delta}_{k} + e_{k+\tau}^{\top}\Gamma_{k+\tau}e_{k+\tau} + \tilde{V}_{k+1}  \mid \IC_{k}^{e} ] \\ 
 &\le&  \Eb[\theta\tilde{\delta}_{k}  + e_{k+\tau}^{\top}\Gamma_{k+\tau}e_{k+\tau} + \bar{V}_{k+1}  \mid \IC_{k}^{e} ] \\ 
 &\le& \Eb[\theta \bar{\delta}_{k} + e_{k+\tau}^{\top}\Gamma_{k+\tau}e_{k+\tau} + \bar{V}_{k+1} \mid   \IC_{k}^{e} ] = \bar{V}_{k},
\end{eqnarray*}
where $\bar{\delta}_{k} = \bar{\gamma}(\IC_{k}^{e})$. The first and second equalities result from backward induction, the first inequality is from the induction hypothesis and the second inequality is from the definition of the sub-optimal triggering policy $\tilde{\gamma}$. 
\end{proof}
The scheduling policy \eqref{eq:aproxy VoI} based on VoI proxy function \eqref{eq:aproxy VoI expression} given in Proposition \ref{proposition:sub-optimal} is no longer optimal for the optimization problem \eqref{eq:original opt}. However, in the following Corollary~\ref{corollary:AoI}, we show that the VoI proxy metric will be shown that the VoI proxy metric outperforms the periodic scheduling policy and AoI-based scheduling policy. 
\begin{remark}
Compared with previous VoI function in \cite{soleymani2021value}, which is of the form
\begin{eqnarray*}
    {\rm VoI}_{k} = e_{k}^{\top}A^{\top}\Gamma_{k+1}A_{k}e_{k} - \theta_{k} + \rho_{k},
\end{eqnarray*} 
the VoI function \eqref{eq:VoI expression} and VoI proxy function \eqref{eq:aproxy VoI expression} is reparameterized with the Gaussian noise and time information which are available at the decision maker at the sensor side. Moreover, they capture the essential variables such as delay and analytically characterize the relationship between communication delay and the formulated optimization problem \eqref{eq:original opt}. These observations facilitate the VoI metric to be applied in a large-scale NCS.

\end{remark}
The result of performance guarantee in Proposition \ref{proposition:sub-optimal} also applies for comparison between VoI proxy-based  triggering policy \eqref{eq:aproxy VoI} and AoI-based triggering policy proposed in \cite{soleymani2021value}. This introduces our next corollary.
\begin{corollary}\label{corollary:AoI}
Let $\hat{\gamma}$ be the AoI-based  triggering policy given
in \cite{soleymani2021value}, under the fixed optimal control law as in \eqref{eq:controller} of Lemma \ref{lemma:certainty eq ctrl}, the sub-optimal triggering policy $\tilde{\gamma}$ given in \eqref{eq:aproxy VoI} outperforms the AoI-based triggering policy $\hat{\gamma}$ in optimization problem \eqref{eq:original opt}.
\end{corollary}
\begin{proof}
Let the cost-to-go function under AoI-based scheduling policy be $\hat{V}_{k}$. We need to prove $ \Psi(\tilde{\gamma}, f^{\ast}) \le \Psi(\hat{\gamma}, f^{\ast})$. Assume that the claim holds for $k+1$, we have 
\begin{eqnarray*}
\tilde{V}_{k} \eq \Eb[\theta\tilde{\delta}_{k} + e_{k+\tau}^{\top}\Gamma_{k+\tau}e_{k+\tau} + \tilde{V}_{k+1} \mid \IC_{k}^{e} ] \\ 
 &\le&  \Eb[\theta\tilde{\delta}_{k} + e_{k+\tau}^{\top}\Gamma_{k+\tau}e_{k+\tau} + \hat{V}_{k+1}  \mid \IC_{k}^{e} ] \\ 
 &\le& \Eb[\theta \hat{\delta}_{k} + e_{k+\tau}^{\top}\Gamma_{k+\tau}e_{k+\tau} + \hat{V}_{k+1} \mid   \IC_{k}^{e} ] = \hat{V}_{k},
\end{eqnarray*}
where $\hat{\delta}_{k} = \hat{\gamma}(\IC_{k}^{e})$. The first and second equalities come from backward induction, and the derivation explanation is same as in the proof of Proposition.~\ref{proposition:sub-optimal}. 
\end{proof}

\section{Simulation Result}\label{sec:simulation}
In this section, we provide a numerical simulation to validate the effectiveness of our proposed VoI-based scheduling policy. In terms of dynamics \eqref{eq:plant}, we choose a two-dimension linear system with system matrices $A = {\rm diag} \{1.15,1.1\}$, $B =  {\rm diag} \{0.1, 0.1\}$, and the Gaussian noise covariance is $W =  {\rm diag}\{0.001, 0.0001\}$. The constant delay  is set as $\tau = 2$. 
Regarding the weighting coefficients in LQG cost function \eqref{eq:original opt}, we choose $Q_{k} =Q_{T+1}= {\rm diag} \{1,1\}$, and $R_{k} =R_{T+1}= {\rm diag} \{1,1\}$ for all $k$. The time horizon is chosen as $T=150$.  The simulated average transmission rate is defined as $r  = \frac{1}{T+1}\sum_{t=0}^{T}\delta_{k} $, and the control performance is valued by the average mean square error $ \frac{1}{T+1}\sum_{t=0}^{T}\|e_{k}\|^{2}$. We run Monte Carlo simulations with 1000 trials. 
\begin{figure}[htbp]
    \centering
    \includegraphics[width=0.45\textwidth]{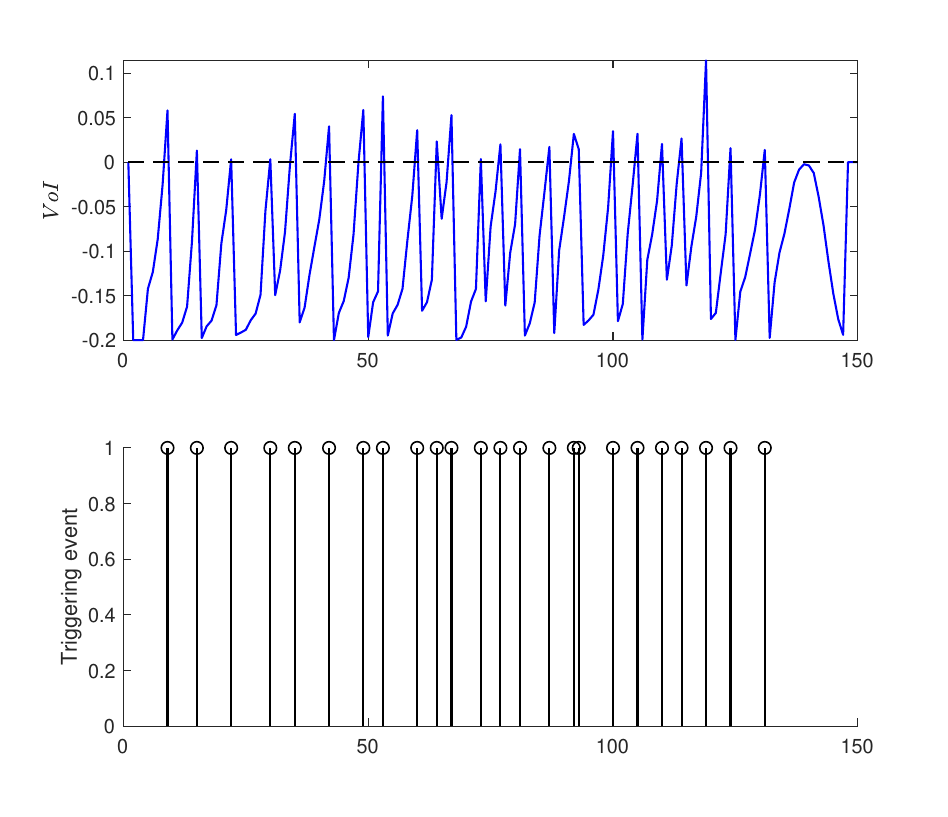}
    \caption{From up to bottom: the trajectory of VoI proxy function \eqref{eq:aproxy VoI expression} and the triggering events of the system.} \label{fig:VoI}
\end{figure} 
\begin{figure}[htbp]
    \centering
    \includegraphics[width=0.45\textwidth]{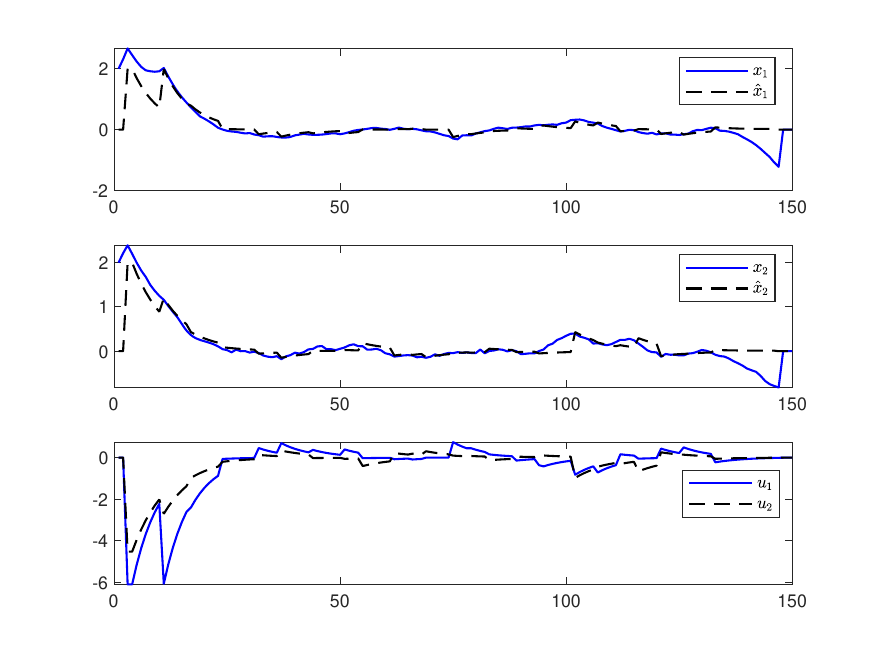}
    \caption{System states and control signals under VoI proxy-based scheduling policy \eqref{eq:aproxy VoI}.} \label{fig:System dynamics}
\end{figure}  
\begin{figure}[htbp]
    \centering
    \includegraphics[width=0.45\textwidth,height =0.37\textwidth
    ]{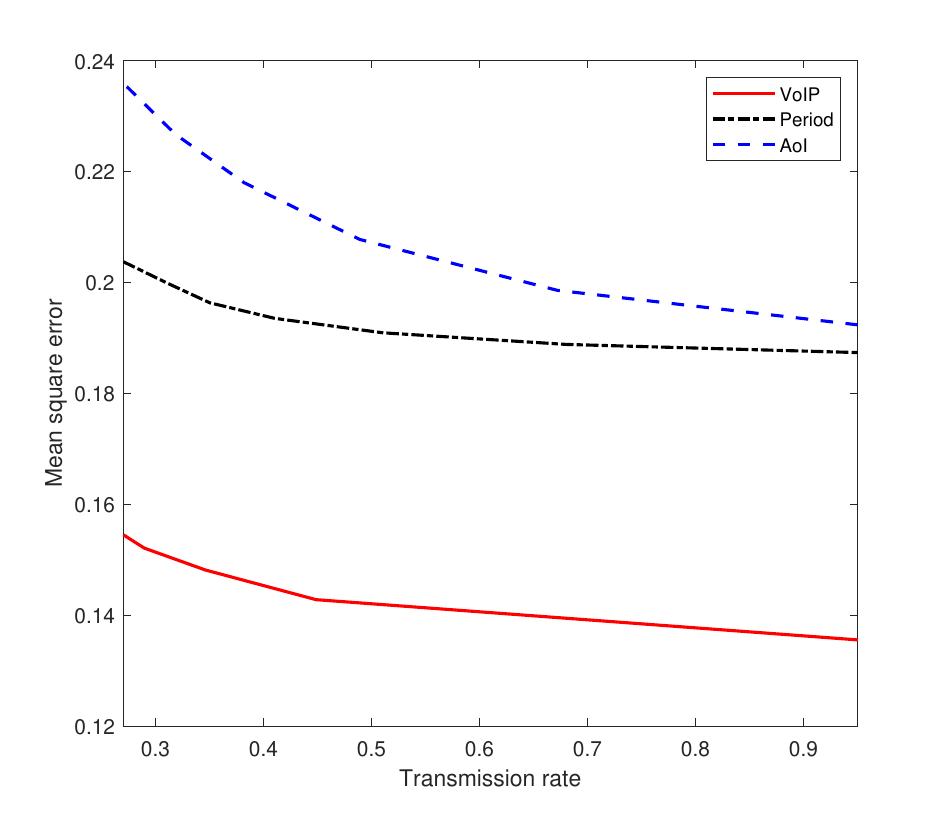}
    \caption{Trade-off between the control performance and transmission rate.} \label{fig:Trade-off}
\end{figure}

Fig.~\ref{fig:VoI} depicts the VoI proxy trajectory \eqref{eq:aproxy VoI expression} and triggering events of the networked control system under the transmission cost $ \theta  = 0.3 $. The triggering occurs only when the VoI function is positive. Under the VoI proxy-based scheduling policy \eqref{eq:aproxy VoI}, the system states  $[x_{1}$ $x_{2}]^{\top}$, their estimation $[\hat{x}_{1}$ $\hat{x}_{2}]^{\top}$ and the optimal LQG control actions  $[u_{1}$, $u_{2}]^{\top}$ are shown in Fig.~\ref{fig:System dynamics}. \\
\indent 
In Fig.~\ref{fig:Trade-off}, we compare the trade-offs between control performance and transmission rate of system under different scheduling policies. They are VoI proxy-based \eqref{eq:aproxy VoI}, periodical and AoI-based scheduling policies \cite{soleymani2021value}, respectively. 
The transmission rate under VoI-based triggering policy decreases with the increasing transmission cost $ \theta $, therefore we obtained the trade-off curve of the control performance with respect to the transmission rates. In order to obtain the same transmission rate range, the transmission costs $\theta$ is chosen from $ 0$ to $0.3$ with the step size of $0.03$. In periodical triggering policy case, the periods are chosen as $T=\{2, 3, 4, 5,6,7\}$. It can be observed that the VoI proxy-based triggering policy leads to a lower average mean square error under the same transmission rates compared with the rest scheduling policies. 


\section{CONCLUSIONS}\label{sec:conclusion}
In this paper, by addressing the trade-off between control performance and communication resource consumption for a NCS, we analytically characterized the relationship between quality of control and VoI function. The derived VoI functions properly reflect the relevance of information including temporal aspects for the control task and is parameterized by network coupling variables such as transmission delay.  
The data packet is transmitted through the network whenever the value of information is positive to preserve the control tasks. Finally, numerical simulation is provided to verify the effectiveness of the proposed VoI-based scheduling policy.

\addtolength{\textheight}{-12cm}   


\begin{thebibliography}{10}

\bibitem{antsaklis2007special}
P.~Antsaklis and J.~Baillieul, ``Special issue on technology of networked
  control systems,'' {\em Proceedings of the IEEE}, vol.~95, no.~1, pp.~5--8,
  2007.

\bibitem{walsh2001scheduling}
G.~C. Walsh and H.~Ye, ``Scheduling of networked control systems,'' {\em IEEE
  Control System Magazine}, vol.~21, no.~1, pp.~57--65, 2001.

\bibitem{park2017wireless}
P.~Park, C.~E. Sinem, C.~Fischione, C.~Lu, and K.~H. Johansson, ``Wireless
  network design for control systems: A survey,'' {\em IEEE Communications
  Surveys Tutorials}, vol.~20, no.~2, pp.~978--1013, 2017.

\bibitem{persis2018power}
D.~C. Persis, W.~Weitenberg, and F.~D\"{o}rfler, ``A power consensus algorithm
  for dc microgrids,'' {\em Automatica}, vol.~89, pp.~364--375, 2018.

\bibitem{bullo2009distributed}
F.~Bullo, J.~Cortés, and S.~Martinez, {\em Distributed control of robotic
  networks: a mathematical approach to motion coordination algorithms}.
\newblock Princeton University Press, 2009.

\bibitem{scholz2009modelling}
B.~Scholz-Reiter, M.~Görges, T.~Jagalski, and A.~Mehrsai, ``Modelling and
  analysis of autonomously controlled production networks,'' {\em IFAC
  Proceedings Volumes}, vol.~42, no.~4, pp.~846--851, 2009.

\bibitem{lunze2010state}
J.~Lunze and D.~Lehmann, ``A state-feedback approach to event-based control,''
  {\em Automatica}, vol.~46, no.~1, pp.~211--215, 2010.

\bibitem{wang2011event}
X.~Wang and M.~D. Lemmon, ``Event-triggering in distributed networked control
  systems,'' {\em IEEE Transactions on Automatic Control}, vol.~56, no.~3,
  pp.~586--601, 2011.

\bibitem{heemels2012introduction}
W.~Heemels, K.~H. Johansson, and P.~Tabuada, ``An introduction to
  event-triggered and self-triggered control,'' in {\em Proc. of 51st {IEEE}
  Conference on Decision and Control (CDC)}, pp.~3270--3285, 2012.

\bibitem{heemels2010networked}
W.~M.~H. Heemels, A.~R. Teel, N.~Van~de Wouw, and D.~Ne{\v{s}}i{\'c},
  ``Networked control systems with communication constraints: Tradeoffs between
  transmission intervals, delays and performance,'' {\em IEEE Transactions on
  Automatic control}, vol.~55, no.~8, pp.~1781--1796, 2010.

\bibitem{yue2013delay}
D.~Yue, E.~Tian, and Q.-L. Han, ``A delay system method for designing
  event-triggered controllers of networked control systems,'' {\em IEEE
  Transactions on Automatic Control}, vol.~58, no.~2, pp.~475--481, 2013.

\bibitem{liu2019survey}
K.~Liu, A.~Selivanov, and E.~Fridman, ``Survey on time-delay approach to
  networked control,'' {\em Annual Reviews in Control}, vol.~48, pp.~57--79,
  2019.

\bibitem{kosta2017age}
A.~Kosta, N.~Pappas, and V.~Angelakis, ``Age of information: A new concept,
  metric, and tool,'' {\em Foundations and Trends in Networking}, vol.~12,
  no.~3, pp.~162--259, 2017.

\bibitem{champati2019performance}
P.~J. Champati, H.~M. Mamduhi, K.~H. Johansson, and J.~Gross, ``Performance
  characterization using aoi in a single-loop networked control system,'' {\em
  INFOCOM Workshops}, pp.~197--203, 2019.

\bibitem{ayan2019age}
O.~Ayan, M.~Vilgelm, M.~Kl{\"u}gel, S.~Hirche, and W.~Kellerer,
  ``Age-of-information vs. value-of-information scheduling for cellular
  networked control systems,'' in {\em Proc. of the 10th ACM/IEEE International
  Conference on Cyber-Physical Systems}, pp.~109--117, 2019.

\bibitem{soleymani2021value}
T.~Soleymani, J.~S. Baras, and S.~Hirche, ``Value of information in feedback
  control: Quantification,'' {\em IEEE Transactions on Automatic Control},
  2021.

\bibitem{molin2019scheduling}
A.~Molin, H.~Esen, and K.~H. Johansson, ``Scheduling networked state estimators
  based on value of information,'' {\em Automatica}, vol.~110, 2019.

\bibitem{molin2013on}
A.~Molin and S.~Hirche, ``On the optimality of certainty equivalence for
  event-triggered control systems,'' {\em IEEE Transactions on Automatic
  Control}, vol.~58, no.~2, pp.~470--475, 2013.

\bibitem{molin2014price}
A.~Molin and S.~Hirche, ``Price-based adaptive scheduling in multi-loop control
  systems with resource constraints,'' {\em IEEE Transactions on Automatic
  Control}, vol.~59, no.~12, pp.~3282--3295, 2014.

\bibitem{ramesh2011on}
C.~Ramesh, H.~Sandberg, L.~Bao, and K.~H. Johansson, ``On the dual effecr in
  state-based scheduling of networked control systems,'' in {\em Proc. of the
  2011 American Control Conference (ACC)}, pp.~2216--2221, 2011.

\end{thebibliography}
\end{document}